\documentclass[12pt]{article}
\usepackage{cite}
\usepackage{amsmath,amssymb,amsfonts,mathtools,amsthm}
\usepackage{graphicx}
\usepackage{nicematrix}
\usepackage{tabularx}
\usepackage{float}
\usepackage[font=footnotesize]{caption}

\newtheorem{theorem}{Theorem}
\newtheorem{assumption}{Assumption}
\newtheorem{proposition}{Proposition}
\newtheorem{lemma}{Lemma}
\newtheorem{remark}{Remark}

\newtheorem{definition}{Definition}
\newtheorem{example}{Example}
\newtheorem{problem}{Problem}
\allowdisplaybreaks

\newcommand{\vect}[1]{\boldsymbol{#1}}
\newcommand{\mat}[1]{\boldsymbol{#1}}
\renewcommand{\eqref}[1]{Eq.~(\ref{#1})}  
  \usepackage{pgfplots,subfig}
\pgfplotsset{ 
  compat=newest, 
}
\usepackage{tikz}
\usepgfplotslibrary{fillbetween}
\usetikzlibrary {positioning}
\pgfplotsset{compat=newest, every tick label/.append style={font=\footnotesize},
every axis label/.append style={font=\footnotesize},  legend style={
             font=\footnotesize},}

\begin{document}

\title{Equilibrium Selection in Replicator Equations Using Adaptive-Gain Control}

\author{Lorenzo Zino$^1$, Mengbin Ye$^2$, Giuseppe C. Calafiore$^{1,3}$,\\and Alessandro Rizzo$^{1,4}$}

\date{\normalsize $^1$ Department of Electronics and Telecommunications, Politecnico di Torino, Torino, Italy\\
$^2$Centre for Optimisation and Decision Science, Curtin University, Perth, Australia\\
$^3$College of Engineering \& Computer Science, VinUniversity, Hanoi, Vietnam\\
$^4$Institute for Invention, Innovation, and Entrepreneurship, New York University Tandon School of Engineering, Brooklyn NY, US}

\maketitle

\begin{abstract}

In this paper, we deal with the {equilibrium selection problem}, which amounts to steering a population of individuals engaged in strategic game-theoretic interactions to a desired collective behavior. In the literature, this problem has been typically tackled by means of open-loop strategies, whose applicability is however limited by the need of accurate a priori information on the game and scarce robustness to uncertainty and noise. Here, we overcome these limitations by adopting a closed-loop approach using an adaptive-gain control scheme within a replicator equation ---a nonlinear ordinary differential equation that models the evolution of the collective behavior of the population. For most classes of $2$-action matrix games we establish  sufficient  conditions to design a controller that guarantees convergence of the replicator equation to the desired equilibrium, requiring limited a-priori information on the game. Numerical simulations corroborate and expand our theoretical findings.
\end{abstract}

\section{Introduction}
\label{sec:introduction}
Evolutionary game theory is a popular mathematical framework, developed to predict the emergent behavior of large populations of  individuals who repeatedly engage in strategic interactions~\cite{Sandholm2010,Hofbauer2009}. Originally proposed to study competition in biological systems~\cite{Smith1973}, it has found applications in a broad range of domains thanks to its flexibility, from  social science to energy and infrastructure networks~\cite{marden2009game,Jiang2014,Barreiro-Gomez2016,Como2021,Stella2022cooperation,Scarabaggio2022,Martins2023}. Within this framework, particular interest has been devoted to \emph{population games}~\cite{Sandholm2010}. In this paradigm, it is assumed that each individual in the population plays a two-player game against all others. Then, based on the game's outcome, they revise their action according to a decision-making process, which is captured at the population-level by way of a \emph{revision protocol}. Technically, revision protocols consist of nonlinear ordinary differential equations (ODEs)~\cite{taylor1978replicator,Hofbauer2009,Sandholm2010} that encapsulate the information of the two-player game and describe the emergent behavior of the population in terms of the evolution of the fraction of adopters of each action over time. Among these revision protocols, one of the most widely used and studied is the \emph{replicator equation}~\cite{taylor1978replicator,Sandholm2010,Cressman2014}. 

For such dynamics, it has been proved that the population converges to a Nash equilibrium (NE)~\cite{Sandholm2010}. However, this NE does not always coincide with the best or the most desirable outcome ---throughout this paper, we use `desirable' from the perspective of a policymaker interested in the collective behavior of the population. For instance, game-theoretic dynamics have been adopted to model human decision-making concerning the adoption of new behaviors and social norms~\cite{montanari2010spread_innovation,kreindler2014rapid_diffusion,ye2021nat}. In this context, the game has typically multiple NEs, and the replicator equation can explain how peer pressure could hinder social change. Moreover, even when the game has a single NE, this may reflect an undesired outcome. This is the case of the prisoner's dilemma, for which  the replicator equation leads to the unique NE of the game, which is however a Pareto-inefficient equilibrium in which all players defect~\cite{Cressman2014}.

Motivated by this wide range of applications, the systems and control community has shown a growing interest in designing and studying individual-level interventions to steering the population to a desired equilibrium~\cite{Quijano2017,riehl2018survey,Grammatico2017}. Such a problem, which is often referred to as \emph{equilibrium selection}~\cite{Kim1996}, has several real-world applications. For instance, in the context of adoption of social norms~\cite{montanari2010spread_innovation,kreindler2014rapid_diffusion,ye2021nat}, an equilibrium selection problem arises whenever a government or a public authority wants to incentivize social change (e.g., toward the adoption of sustainable behavior)~\cite{Fan2020,Selin2021}; in social dilemmas, when a policymaker wants to promote cooperation~\cite{May1981,Stella2022cooperation}.

Several efforts have been made to address the equilibrium selection problem. Most of the approaches rely on open-loop schemes and can be classified into three main families. The first family includes methods that rely on directly controlling the state of some individuals of the population who behave as influencers or committed minority, helping to steer the rest of the population to the desired equilibrium~\cite{morris2000contagion,centola2018experimental_tipping,Como2022targeting}. The second family, instead, consists of methods that impact an individual's decision-making, by incorporating additional behavioral mechanisms, such as enhancing an individual's sensitivity to trends~\cite{ye2021nat,cdc2021,Zino2022nexus}, enforcing reciprocity~\cite{vanVeelen2012,Park2022cooperative}, or acting on social reputation~\cite{Giardini2021}. However, such direct individual-level interventions are not always feasible in the real-world, due to the limited possibility to directly influence an individual's state or decision-making mechanisms. To address these limitations, a third family of approaches has been proposed, in which the control is exerted in an indirect fashion by adjusting the structure of the two-player game played by the population. This captures, e.g., the implementation of incentives to favor a desired action against the others. Such adjustments, however,  require complete information on the  game~\cite{riehl2018incentive,Eksin2020,Gong2022,Zhu2023} or entail nontrivial optimization processes if implemented dynamically~\cite{MartinezPiazuelo2023}, calling for the development of more flexible closed-loop control schemes, which are able to solve the equilibrium selection problem by dynamically adjusting the structure of the two-player game, requiring limited a priori information.

In this paper, we fill in this gap by proposing a novel feedback control scheme to solve the equilibrium selection problem. Specifically,  we adopt an adaptive-gain control scheme, motivated by its successful application to complex systems in other contexts~\cite{ioannou1996robust,yu2012distributed_adaptive,mei2016adaptive,Walsh2023_IFAC}. In particular, we  focus on the broad class of symmetric two-player two-action matrix games~\cite{riehl2018survey}, where each player can chose among two possible actions (e.g., adopting or not adopting a new sustainable behavior in the context of social change; or cooperate vs. defect in social dilemmas), and the structure of the game is captured by a \emph{payoff matrix}. Our proposed approach involves using a gain to control one entry of the payoff matrix, and this gain adaptively changes following an algorithm that takes information from the population state, creating a feedback loop between the controller and the population. Our key objective is to establish sufficient conditions for the design of the adaptive-gain controller to solve the equilibrium selection problem, requiring only limited a priori information on the game. 

Technically, by coupling the replicator equation with an adaptively evolving gain, the nonlinear system of ODEs of the controlled evolutionary dynamics increases in dimension as well as complexity. The result is an increased 
complexity of the system and, thus, the range of possible emergent behaviors. Using systems- and control-theoretic tools~\cite{Pachpatte_book,khalil2002nonlinear}, we  analyzed the obtained nonlinear coupled ODEs, establishing easy-to-verify sufficient conditions under which the adaptive-gain controller succeeds in solving the equilibrium selection problem with limited required information on the payoff matrix, thereby addressing many limitations of existing methods.

After formalizing the adaptive-gain controller and proving its well-posedness, we observe that the equilibrium selection problem has three distinct manifestations, depending on the nature of the desired equilibrium: i) enforcing convergence to a locally (but not globally) stable NE, ii) stabilizing an unstable equilibrium; and iii) steering the system to a point that is not an equilibrium of the uncontrolled replicator equation. For the first scenario, which often arises in the context of social change, we establish sufficient conditions for solving the equilibrium selection problem, depending on the entry of the payoff matrix on which the gain is implemented. For the second scenario, which captures, e.g., promoting cooperation in social dilemmas, we prove that the equilibrium selection problem can be solved only by controlling a specific entry of the payoff matrix. If this entry is controlled, then we establish mild conditions that are sufficient to steer the population to the desired equilibrium. Finally, we tailor the adaptive-gain controller to solve the more challenging third scenario for many relevant classes of games, including the anti-coordination games that captures congestion problems in infrastructure networks~\cite{Jiang2014,Como2022traffic}. 

Some results appeared in~\cite{zino2023_adaptive}, in a preliminary form. Here, we expand on that effort along several directions, including i) a more general formulation and extensive characterization of the equilibrium selection problem in terms of three distinct scenarios; ii) the general treatment of the first two scenarios with the establishment of easy-to-implement sufficient conditions to guarantee convergence to the desired equilibrium (this extends our preliminary results in~\cite{zino2023_adaptive}, which only presented a specific solution to the first two problems); iii) the analysis of the third scenario, which is the most challenging and whose formulation and treatment is entirely new, and iv) numerical simulations to explore the effectiveness of the proposed adaptive-gain controller, giving insights beyond our theoretical guarantees. 

The rest of this paper is organized as follows. In Section~\ref{sec:model}, we present the mathematical model. In Section~\ref{sec:problem}, we propose our adaptive-gain control scheme and formalize the equilibrium selection problem. Our main theoretical results are presented in Sections~\ref{sec:appraoch}--\ref{sec:setpoint}. Section~\ref{sec:conclusions} concludes the paper and suggests future research.

\section{Population Games and Replicator Equation}\label{sec:model}


We consider a (large) population where each individual repeatedly plays the same  two-player matrix game~\cite{riehl2018survey} with all the others, dynamically revising their strategy toward increasing their payoff. The emergent behavior of the population can be studied with the methods developed in evolutionary game theory~\cite{Sandholm2010}. Specifically, we employ a replicator equation~\cite{taylor1978replicator}, which allows to capture the emergent behavior of the population by means of an ODE. Our ultimate goal is to leverage such an ODE to design a feedback control scheme to adaptively change the entries of the payoff matrix using population-level information  in a feedback fashion, to steer the population to a desired equilibrium, as illustrated in Fig.~\ref{fig:schematic}.

\begin{figure}
    \centering
\includegraphics[]{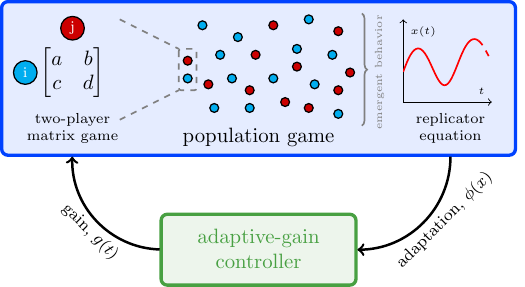}\caption{Schematic of the controlled evolutionary game-theoretic dynamics. }
    \label{fig:schematic}
\end{figure}

In the rest of this section, we introduce the mathematical formalism used thought the paper. Specifically, we start by presenting two-player matrix games and their key properties. Then, we introduce the replicator equation and describe its behavior in the uncontrolled scenario. 

\subsection{Two-player matrix game}

Each individual plays a  symmetric\footnote{A game is symmetric when players share an identical payoff matrix.} two-player matrix game. In this game, each player can choose between two \emph{actions}, termed action $1$ and action $2$, characterized by a payoff matrix 
\begin{equation}\label{eq:payoff}
    {\mat A}=\begin{bmatrix}  a&b\\c&d \end{bmatrix},
\end{equation}
with $a,b,c,d\in\mathbb R$, where $\mathbb R$ is the set of real numbers. In other words, if we consider two players with configuration of actions $\vect{y}=(y_i,y_j)\in\{1,2\}\times\{1,2\}$, then player $i$ would receive payoff equal to $a$ (or $b$) for selecting action $y_i=1$ against an opponent who plays $y_j=1$ (or $y_j=2$). Meanwhile, player~$i$ receives $c$ (or $d$) for selecting action $y_i=2$ against an opponent who plays $y_j=1$ (or $y_j=2$).

An important concept in game theory is that of \emph{Nash equilibrium} (NE), which captures specific configurations such that no player can increase their payoff by unilaterally changing the action played. In the context of symmetric two-player matrix games, we have the following definition of NEs.

\begin{definition}[Nash equilibrium]\label{def:Nash}
In the context of symmetric two-player matrix games,    a \emph{pure NE} is a configuration of actions of the two players $\vect{\bar y}=(\bar y_i,\bar y_j)$ such that $A_{\bar y_i\bar y_j}\geq  A_{s\bar y_j}$ and $A_{\bar y_j\bar y_i}\geq  A_{s'\bar y_i}$, for all $s,s'\in\{1,2\}$. If configurations are allowed to be defined as  probability mass functions over the actions (namely, player $i$ chooses action $1$ with probability $p_i\in[0,1]$ and action $2$ otherwise), then a \emph{mixed NE} is a (probabilistic) configuration $\vect{\bar p}=(\bar p_i,\bar p_j)$ in which no player can increase their expected payoff by unilaterally changing their probability mass function, i.e.,
$\bar p_i\bar p_jA_{11}+\bar p_i(1-\bar p_j)A_{12}+(1-\bar p_i)\bar p_jA_{21}+(1-\bar p_i)(1-\bar p_j)A_{22}\geq q\bar p_jA_{11}+q (1-\bar p_j)A_{12}+(1-q)\bar p_jA_{21}+(1-q)(1-\bar p_j)A_{22}$, for any $q\in[0,1]$, and the same inequality holds for player $j$.
\end{definition}

Following a standard classification for two-player matrix games~\cite{riehl2018survey}, the payoff matrix in \eqref{eq:payoff} determines three different classes of games, summarized in the following.

\begin{proposition}\label{prop:nash}\rm The payoff matrix in \eqref{eq:payoff} determines three classes of games, characterized in terms of their NEs:
\begin{enumerate}
    \item If $d>b$ and $a>c$, we have a \emph{coordination game}, which has two pure NEs $(1,1)$ and $(2,2)$, and a mixed NE, where action $1$ is played by either players with probability equal to
    \begin{equation}\label{eq:mixed}
   x^*:=\frac{d-b}{a+d-b-c}. 
\end{equation}
    \item If $d>b$ and $a<c$, or $d<b$ and $a>c$, we have a \emph{dominant-strategy game}, which has a unique (pure) NE:  $(2,2)$ if $d>b$ and $a<c$, or  $(1,1)$ if $d<b$ and $a>c$.
    \item If $d<b$ and $a<c$, we have an \emph{anti-coordination game}, which has two (pure) NEs $(1,2)$ and $(2,1)$, and a mixed NE, where action $1$ is played with probability equal to $x^*$ from \eqref{eq:mixed} by either players.
\end{enumerate}
\end{proposition}

In the following, we provide one example for each class, which will be used later in the paper to illustrate our findings. %

\begin{example}[Pure coordination game]\label{ex:pure}
A game with diagonal and strictly positive payoff matrix ($b=c=0$, $a,d>0$) is a particular coordination game, called {pure coordination game}, in which a player receives a payoff if and only if they choose the same action of the other player. According to Proposition~\ref{prop:nash}, the game has two pure NEs and a mixed NE at $x^*={d}/{(a+d)}$.
\end{example}

\begin{example}[Prisoner's dilemma]\label{ex:pd}
A game with $c<a<d<b$ is a prisoner's dilemma in which action $1$ and $2$ represent defection and cooperation, respectively. A player receives $d$ for mutual cooperation and $b>d$ for cheating (i.e, to defect if the other cooperates), while $a$ is the punishment for mutual defection, which provides a smaller payoff than mutual cooperation ($d$), but larger than cooperating if the other defects ($c$). Prisoner's dilemma is a dominant-strategy game, with mutual defection $(1,1)$ as the unique NE. This is somewhat counter-intuitive, since the NE is not Pareto optimal: both players would receive a larger payoff if they both cooperate. 
\end{example}

\begin{example}[Minority game]\label{ex:min}
A minority game is a game with zero-diagonal and strictly positive off-diagonal payoff matrix ($a=d=0$, $b,c>0$). This is an anti-coordination game in which players receive a payoff only if they do not coordinate on the same action. The game has two pure NEs at $(1,2)$ and $(2,1)$, and a (unique) mixed NE with $x^*={d}/{(a+d)}$.
\end{example}

\subsection{Replicator equation}

In order to describe the emergent behavior at the population-level, we adopt the standard approach used in population games~\cite{Sandholm2010} of considering the population as a continuum of players of unit mass. Each player repeatedly engages in the two-player matrix game described above with their peers, revising their strategy toward improving their payoff. Specifically, we denote by $x(t)\in[0,1]$ the fraction of adopters of action $1$ at time $t\geq 0$; consequently, the fraction of adopters of action $2$ is equal to $1-x(t)$. Then, the total reward $r_i(x,\mat A)$ associated with action~$i$ is given by the average payoff that a player receives for choosing that action from all games played~\cite{Sandholm2010},  which is equal to
\begin{equation}
   \begin{bmatrix}  r_1(x,{\mat A})\\ r_2(x,{\mat A}) \end{bmatrix}={\mat A}\begin{bmatrix}  x\\1-x \end{bmatrix}= \begin{bmatrix}  ax+b(1-x)\\ cx+d(1-x)\end{bmatrix},
\end{equation}
for action $1$ and $2$, respectively. For the sake of readability, we have omitted to explicitly write the dependence of the variables on time $t$. Throughout this paper, we adopt this convention except when we wish to  highlight such dependence. 

In population games~\cite{Sandholm2010}, the revision of the players' actions is captured at the population-level by means of ODEs. Here, we use the replicator equation~\cite{taylor1978replicator}, yielding:
\begin{equation}\label{eq:replicator}
    \dot x=x(1-x)(r_1(x,{\mat A})-r_2(x,{\mat A}))=x(1-x)((a+d-b-c)x+b-d),
    \end{equation}
with initial condition $x(0)\in[0,1]$. Briefly, \eqref{eq:replicator} captures the tendency of players to imitate their peers who have higher reward: the rate at which individuals switch to action $1$ is proportional to the difference in the reward for playing $1$ with respect to the reward for playing $2$. The (uncontrolled) replicator equation in \eqref{eq:replicator} has always the two \emph{consensus} equilibria $x+0$ and $x=1$, where all individuals adopt the same action. However, the stability of the consensus equilibria and the presence of other equilibria depend on the characteristics of the game. Overall, its behavior, which has been extensively studied in the literature~\cite{taylor1978replicator,Sandholm2010},  is summarized in the following result.

\begin{proposition}\label{prop:convergence_uncontrolled}\rm
Consider the replicator equation in \eqref{eq:replicator}. If the payoff matrix ${\mat A}$ in \eqref{eq:payoff} is
\begin{enumerate}
    \item a coordination game, then $x(t)\to 0$ if $x(0)<x^*$, and $x(t)\to 1$ if $x(0)>x^*$; with $x^*$ from \eqref{eq:mixed};
    \item a dominant-strategy game, and
    \begin{enumerate}
        \item  
 $d>b$ and $a<c$, then $x(t)\to 0$ for any $x(0)<1$;
    \item $d<b$ and $a>c$, then $x(t)\to 1$ for any $x(0)>0$;
       \end{enumerate}
    \item  an anti-coordination game, then $x(t)\to x^*$ for any $x(0)\in(0,1)$, with $x^*$ from \eqref{eq:mixed}.
\end{enumerate}
\end{proposition}
\begin{proof}First, observe that $[0,1]$ is invariant under \eqref{eq:replicator}, and hence all trajectories must converge to an equilibrium. Then, observe that \eqref{eq:replicator} has at most three equilibria: the pure configurations $x=1$ and $x=0$, in which the entire population adopts action $1$ and $2$, respectively; and  an additional equilibrium $x^*=\frac{d-b}{a+d-b-c}$, which belongs to the domain $[0,1]$ and is distinct from $x=1$ and $x=0$ if and only if i) $d>b$ and $a>c$, or ii) $d<b$ and $a<c$. In this case, $x^*$ coincides with the mixed NE of the two-player game. Finally, the analysis of the sign of $\dot x$ in \eqref{eq:replicator} yields the claim.
\end{proof}

\begin{figure}
    \centering
    \subfloat[Coordination]{\includegraphics[]{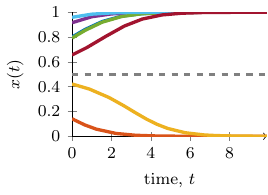}\label{fig:2a}}\quad
    \subfloat[Dominant-strategy]{\includegraphics[]{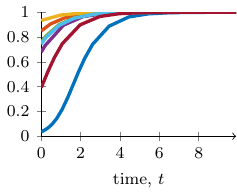}\label{fig:2b}}\quad
    \subfloat[Anti-coordination]{\includegraphics[]{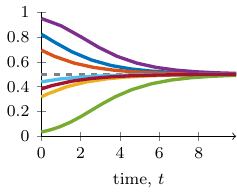}\label{fig:2c}}
\caption{Trajectories of the (uncontrolled) replicator equation in \eqref{eq:replicator} for (a) a pure coordination game (Example~\ref{ex:pure}) $a=d=1$; (b) a prisoner's dilemma (Example~\ref{ex:pd}) with $c=0$, $a=1$, $d=2$, and $b=3$; and (c) a minority game (Example~\ref{ex:min}) with $b=c=1$. Different colors denote different initial conditions. }
    \label{fig:2}
\end{figure}

\begin{remark}\label{rem:equilibria}
The behavior of \eqref{eq:replicator} can be related to the NEs of the corresponding two-player game. In fact, all the NEs of \eqref{eq:payoff} correspond to equilibria of \eqref{eq:replicator}. However, \eqref{eq:replicator} may have (unstable) equilibria that are not NE ---e.g., $x=1$ in scenario 2). Moreover, not all NEs are asymptotically stable ---see, e.g., the mixed NE in scenario 1). More precisely, the equilibria that are asymptotically stable coincide with the evolutionarily stable strategies of the game~\cite{Sandholm2010,Smith1973}. 
\end{remark}

Figure~\ref{fig:2} offers graphical insights into the results summarized in Proposition~\ref{prop:convergence_uncontrolled}. For coordination games (Fig.~\ref{fig:2a}), the (uncontrolled) dynamics converges to a pure configuration that depends on whether the initial condition is above or below the (unstable) mixed NE $x^*$ (represented in the figure by a by a gray dashed line). For dominant-strategy game, such as the prisoner's dilemma in the simulations in Fig.~\ref{fig:2b}, we observe convergence to the all-defector configuration from any (interior) initial condition. Finally, for anti-coordination games, Fig.~\ref{fig:2c} shows convergence to the unique (mixed) NE (gray dashed line), for any (interior) initial condition.

\section{Closed-loop Control Scheme}\label{sec:problem}

We take the perspective of a central policymaker, whose objective is to design an intervention policy that aims to control the emergent behavior of the population (whose dynamics are described by the replicator equation in \eqref{eq:replicator}), steering the population to a desired equilibrium point $x=\bar x$. The intervention acts at the individual-level, by changing the entries of the payoff matrix in \eqref{eq:payoff}, e.g. via incentives. Clearly, an intuitive intervention policy would be to act in a open-loop fashion by permanently changing the payoff matrix to shift the system to resemble a scenario where the desired equilibrium $x=\bar x$ is globally asymptotically stable (using Proposition~\ref{prop:convergence_uncontrolled}). However, this approach has several limitations. First, in many real-world scenarios it may not be feasible, since permanent changes to the payoff structure may be impossible to be implemented or economically unsustainable. Second, in some scenarios it may not be necessary to permanently change the payoff, since interventions might be needed only in a transient phase, or only in some specific conditions. Third, the aforementioned approaches would require exact knowledge of the values in the payoff matrix, which may not be available in many real-world applications or could be subject to noise. For these reasons, we focus on designing control schemes that rely on a closed-loop and adaptive structure, using information at the population-level, as illustrated in Fig.~\ref{fig:schematic}.

\subsection{Adaptive-gain controller}

We assume that we can control one entry of the payoff matrix through an additive gain. The payoff matrix becomes
    \begin{equation}\label{eq:payoff_control}
    {\mat A}(t):= \mat{\hat A}+{\mat G}g(t)=\begin{bmatrix}   a & b\\ c& d \end{bmatrix}+\mat{G}g(t),
\end{equation}
where $\mat{\hat A}$ is the \emph{nominal payoff matrix} of the uncontrolled game from \eqref{eq:payoff}, $g(t):\mathbb R_{+}\to\mathbb R_{+}$ is a continuous function mapping non-negative real numbers into non-negative real numbers that quantifies the \emph{control gain}, and $\mat G$ is a \emph{control matrix} that determines which entry of the payoff is controlled and is selected among the following four matrices:
\begin{equation}\label{eq:matrices}
    \begin{array}{l}
    \mat{G^{(1)}}:=\begin{bmatrix}
    1&0\\0&0
    \end{bmatrix},\,\,   \mat{G^{(2)}}:=\begin{bmatrix}
    0&1\\0&0
    \end{bmatrix},\,\,\mat{G^{(3)}}:=\begin{bmatrix}
    0&0\\1&0
    \end{bmatrix},\,\,   \mat{G^{(4)}}:=\begin{bmatrix}
    0&0\\0&1
    \end{bmatrix}.
\end{array}
\end{equation}

\begin{remark}Theoretically, a controller may act on multiple entries of the payoff matrix, i.e., with a linear combination of the matrices in \eqref{eq:matrices}. However, using the  control matrices in \eqref{eq:matrices} allows us to factorize the control inputs, clearly understanding the challenges and benefits of acting on each one of the entries. Moreover, for many real-world applications, practical limitations may limit the possibility to control multiple entries.   
Consequently, we focus on solving the equilibrium selection problem using the control matrices in \eqref{eq:matrices}. 
\end{remark}

There are essentially two different classes of control matrices: $\mat{G^{(1)}}$ and $\mat{G^{(4)}}$ adaptively increase the payoff associated with coordination on action $1$ or $2$, respectively; whereas $\mat{G^{(2)}}$ and $\mat{G^{(3)}}$ provide an  advantage for choosing one action against a player who plays the other action ($1$ against $2$ and $2$ against $1$, respectively). From a practical viewpoint, these two classes of matrices reflect two opposite types of interventions: the former captures intervention policies aiming at favoring the collective adoption of the desired action; the latter, instead, captures interventions incentivizing innovators and early adopters of the desired action.  For instance, consider a policymaker who wishes to promote sustainable mobility by having people move away from using cars. The first class captures interventions such as offering joint discounted tickets for traveling by public transportation or advantages for sharing cars; the second class captures giving benefits or discounts for those who go to work using bikes or public transportation. In view of these observations, we introduce the following classification.

\begin{definition}\label{ref:gain}
    We refer to \eqref{eq:payoff_control} as
    \begin{itemize}
        \item a \emph{conformity gain controller} if we use $\mat{G^{(1)}}$ or $\mat{G^{(4)}}$; 
        \item an \emph{innovation gain controller} if we use $\mat{G^{(2)}}$ or $\mat{G^{(3)}}$.
    \end{itemize} 
\end{definition}
    
Here, we propose to design the control gains to obtain a closed-loop controller that uses information on the dynamics at the population-level, i.e., $x(t)$, to adaptively change the selected entry of the payoff matrix, as illustrated in Fig.~\ref{fig:schematic}. Hence, we propose that the gain $g(t)$ is governed by the ODE:
\begin{equation}\label{eq:gain}
  \dot g(t)=\phi(x(t))g(t),
\end{equation}
where the \emph{adaptation rate} $\phi(x):[0,1]\to \mathbb R$ is continuously differentiable on its domain, with initial condition $g(0)>0$. As such, while $\mat G$ determines the type of intervention (coordination vs innovation), $g(t)$ represents the (adaptively evolving) amount of incentive.

\begin{remark}
    A possible interpretation for \eqref{eq:payoff_control} and \eqref{eq:gain} is that of incorporating a control action within a payoff dynamics model for evolutionary game dynamics~\cite{Park2019,Arcak2021,MartinezPiazuelo2022}. 
\end{remark}

In summary, our class of adaptive-gain controllers is characterized by the pair $(\mat{G},\phi)$ of control matrix and adaptation rate. The controlled replicator equation can be formulated as the following planar system of coupled nonlinear ODEs:
\begin{equation}\label{eq:controlled_replicator}
\begin{alignedat}{2}  
    \dot x&=&&x(1-x)\big((a+d-b-c)x+b-d\\&&&+(G_{11}-G_{21})gx+(G_{12}-G_{22})g(1-x)\big),\\
    \dot g&=&&\phi(x)g,
    \end{alignedat}
\end{equation}
which is well posed, as proved in the following.

\begin{lemma}\label{lemma:invariance}
The domain $\Omega:=[0,1]\times \mathbb R_{+}$ is positively invariant under \eqref{eq:controlled_replicator}.
\end{lemma}
\begin{proof}
The regularity of $\phi$ and the presence of the term $g$ on the right-hand side of the second equation of \eqref{eq:controlled_replicator} guarantees that $g(t)\in\mathbb R_{+}$ for all $t\in\mathbb R_{+}$. This observation, combined with the regularity of the right-hand side of the first equation of \eqref{eq:controlled_replicator} and the fact that it is equal to $0$ at the boundaries $x=0$ and $x=1$ yields the claim~\cite{blanchini1999set}.
\end{proof}

\subsection{Problem formulation}

Proposition~\ref{prop:convergence_uncontrolled} and Remark~\ref{rem:equilibria} shed light on important aspects of the replicator equation, relevant for  the equilibrium selection problem. In particular, they highlight different potential manifestations of the problem, depending on whether the objective is to steer the system to an equilibrium that i) is locally stable for the uncontrolled dynamics, ii) is an unstable consensus equilibrium, or iii) is a mixed strategy state, which need not be an equilibrium for the original dynamics. These different manifestations of the equilibrium selection problem share many similarities, but also have some key differences, calling for different treatments.

In the first scenario, the uncontrolled dynamical system has multiple equilibria that are locally ---but not globally--- stable and one may want to design a controller to steer the system to a desired equilibrium, regardless of the initial condition, making it (almost) globally asymptotically stable for the controlled system. In the following, we will refer to this as the \emph{consensus reaching} problem. This problem naturally arises in  coordination games (see Fig.~\ref{fig:2a}), where the 
 uncontrolled dynamics has two locally asymptotically stable equilibria ($x=0$ and $x=1$) and the initial condition determines which equilibrium is reached. In the context of social change (where action~$1$ and $2$ represent status quo and innovation, respectively), a key problem is to design incentive schemes to favor social change when innovators (players adopting action~$2$) are initially a small minority~\cite{bass1969model,ye2021nat}, i.e., when the initial condition is out of the basin of attraction of the desired equilibrium.

In the second scenario, instead, the desired consensus equilibrium is unstable for the uncontrolled dynamics. Hence, the objective of the controller becomes to enforce convergence to such desired equilibrium. We will refer to this as the \emph{consensus stabilization} problem. 
Consensus stabilization problems may arise in dominant-strategy and anti-coordination games, where the uncontrolled system has a unique globally asymptotically stable equilibrium (see Figs.~\ref{fig:2b} and~\ref{fig:2c}), but the controller wants to steer the system to a different (unstable) equilibrium, e.g., to promote cooperation in the prisoner's dilemma (Example~\ref{ex:pd})~\cite{May1981,Stella2022cooperation}, or reach coordination in anti-coordination games (Example~\ref{ex:min})~\cite{ramazi2016networks}. 

Finally, the third scenario consists of steering the system to a desired mixed strategy equilibrium, which in general is not an equilibrium of the uncontrolled dynamics. This problem is common in the control literature, in particular in process control. As such, we shall refer to this as the \emph{set-point regulation} problem. Applications of interest for such a problem are typical of many socio-technical systems, where incentives should be designed in order to prevent the entire population from adopting the same action, which may lead to congestion in infrastructure networks~\cite{Jiang2014,Como2022traffic} or reduce the resilience of financial networks~\cite{Acemoglu2015}. 

Besides guaranteeing convergence to the desired equilibrium, it also makes sense for many practical applications to enforce that the gain converges to a constant value. Here, a key difference between the three manifestations of the equilibrium selection problem arises. In fact,  for the consensus reaching problem, once the system is in the basin of attraction of the desired equilibrium, then no control is needed to guarantee convergence. On the contrary, in the consensus stabilization and set-point regulation problems, the controller should enforce convergence to an equilibrium that would otherwise be unstable and thus it should remain active at steady state, as proved by the following impossibility result.

\begin{lemma}\label{lemma:g}
If $x=\bar x$ is not a stable equilibrium of \eqref{eq:replicator} or the unstable mixed strategy NE of a coordination game, then it is impossible to guarantee convergence to $\bar x$ from all initial conditions in $(0,1)$ using an adaptive-gain controller of the form $(\mat G,\phi)$ such that $g(t)\to 0$. 
\end{lemma}
\begin{proof}
From Proposition~\ref{prop:convergence_uncontrolled}, we observe that in order for $x=\bar x$ to not be a stable equilibrium of \eqref{eq:replicator}, three cases are possible. Specifically,  $\bar x$  is either i) an unstable consensus equilibrium in a dominant strategy or an anti-coordination game, ii) an interior non-equilibrium point (i.e., any interior point excluding the stable mixed NE of an anti-coordination game and the unstable mixed NE of a coordination game), or {\color{black}iii) an unstable mixed NE of a coordination game. }
    
In case i), consider $\bar x=0$ when $b-d>0$. Since $b-d>0$, then, by continuity, there exist two positive constants $\delta,\gamma>0$ such that $(a+d-b-c)x+b-d>\gamma$, for all $x\in[0,\delta]$. The hypothesis that $\lim_{t\to\infty}g(t)=0$ implies that there exists $\bar t$ such that $g(t)\leq{\gamma}$ for all $t \geq \bar t$. As a consequence, we can bound $\dot x(t)\geq \frac12\gamma x(1-x)>0$ for any $t\geq\bar t$ and $x\in[0,\delta]$. Note that, regularity of $\phi(x)$ implies that exists a constant $\bar\gamma$ such that $g(t)\geq\bar\gamma$ for any $t\in[0,\bar t]$. Therefore, from \eqref{eq:controlled_replicator} we bound $\dot x\geq -(\max\{G_{12},G_{22}\}\bar\gamma-c-d)x$ for $t\in [0,\bar t$), which implies, due to Gronwall's inequality~\cite{Pachpatte_book}, that $x(\bar t)\geq x(0)\exp\{-(\max\{G_{12},G_{22}\}\bar\gamma-c-d)\bar t\}$ in the same interval. In other words, $x(t)$ cannot converge to $\bar x=0$ before time $\bar t$ and, for any $t\geq \bar t$,  $\dot x(t)$ is strictly positive in a neighborhood $[0,\delta]$ of $\bar x=0$. This latter property implies that $x(t)$ cannot converge to $0$ after $\bar t$. A similar argument can be used for $\bar x=1$ when $c-a>0$.
    
In case ii), $\bar x$ not being an equilibrium implies that $(1-x)x((a+d-b-c)\bar x+b-d)\neq 0$. Without loss of generality, let us assume that the left side of the inequality is strictly positive. By continuity, there exist $\delta,\gamma>0$ such that $(1-x)x((a+d-b-c)x+b-d)>\gamma$, for all $x\in[\bar x-\delta,\bar x+\delta]$. The same argument used in case i) guarantees that there exists $\bar t$ such that for any $x\in[\bar x-\delta,\bar x+\delta]$, $\dot x(t)>0$ for all $t\geq\bar t$. Hence, $\bar x$ cannot be an equilibrium.

In case iii), the Jacobian of the controlled replicator equation at the equilibrium point $(x,g)=(\frac{d-b}{a+d-b-c},0)$ is upper triangular matrix with eigenvalues equal to $\frac{(a-c)(d-b)}{a+d-b-c}>0$ and $\phi(\frac{d-b}{a+d-b-c})<0$. This implies that the desired equilibrium is a saddle point, and thus unstable. \end{proof}%

Consequently, while for the consensus reaching problem it is suitable to consider control policies that are only active in the transient, and $g(t)\to 0$; in the other two scenarios we need the control gain to remain active in the steady state, and $g(t)$ to converge to a finite non-zero constant. Hence, we  formulate three different control design problems in terms of the three different manifestations of the equilibrium selection problem described above.

\begin{problem}[Consensus reaching]\label{pr:1}
    Consider a game, where it is known that an equilibrium $x=\bar x$ is (locally) stable. Design $(\mat{G},\phi)$ so that i) the state $x(t)\to \bar x$ for (almost) any initial condition and ii) the gain $g(t)\to 0$.
\end{problem}
    
\begin{problem}[Consensus stabilization]\label{pr:2}
      Consider a game, where it is known that a consensus equilibrium $x=\bar x$ is unstable. Design $(\mat{G},\phi)$ so that i) the state $x(t)\to \bar x$ for (almost) any initial condition and ii) the gain $g(t)\to \bar g$, for some constant $\bar g>0$.
\end{problem}

\begin{problem}[Set-point regulation]\label{pr:3}
      Consider a game and let $x=\bar x\in(0,1)$. Design $(\mat{G},\phi)$ so that i) the state $x(t)\to \bar x$ for (almost) any initial condition and ii) the gain $g(t)\to \bar g$, for some constant $\bar g>0$.
\end{problem}

Considering the differences between these three problems, we address them separately in the next three sections. In particular, Section~\ref{sec:appraoch} focused on Problem~\ref{pr:1}, where we establish easy-to-verify sufficient conditions on the adaptation rate $\phi(x)$ to solve the problem for both the conformity gain and the innovation gain controllers (see Definition~\ref{ref:gain}). Then,  in Section~\ref{sec:dominant},  we prove that only the conformity gain controller can solve Problem~\ref{pr:2}, under some easy-to-verify conditions on the design of the adaptation rate. For these first two problems, without any loss in generality, we will focus on the scenario in which the desired consensus equilibrium is $\bar x=0$. Finally, in Section~\ref{sec:setpoint}, we explicitly design a class of adaptive-gain controllers to solve Problem~\ref{pr:3} for dominant-strategy and anti-coordination games.

\section{Consensus Reaching Problem}\label{sec:appraoch}

The consensus reaching problem arises in coordination games, where multiple equilibria are locally ---but not globally--- stable. Let us consider a generic coordination game, with payoff matrix from \eqref{eq:payoff_control} with $a>c$ and $d>b$. As a consequence of Proposition~\ref{prop:convergence_uncontrolled}, the trajectories of the uncontrolled dynamics will converge to $x=0$ if $x(0)<x^*=\frac{d-b}{a+d-b-c}$, and to $x=1$ if $x(0)>x^*$, with a saddle point equilibrium at $x^*$. Without loss of generality, we focus on the case in which the desired equilibrium is $\bar x=0$; the case $\bar x=1$ can be treated in a similar manner, as discussed at the end of this section.  Consequently, we make the following assumption.

\begin{assumption}\label{a:delta}
    The uncontrolled replicator equation \eqref{eq:replicator} is a coordination game that admits the (locally) asymptotically stable equilibrium $\bar x=0$. Moreover, the control designer knows the value of a constant $\delta>0$ such that $x(0)\in[0,\delta]\implies x(t)\to0$.
\end{assumption}

In other words, we make the reasonable assumption that a (possibly conservative) estimate of the domain of attraction of the desired equilibrium is known to the control designer. Note that such assumption is quite mild, since, with local information about the equilibrium available (namely, we know that the equilibrium is stable), one can assume that a neighborhood of attraction for the equilibrium can be easily estimated. In real-world scenarios, such a bound can be set using experimental data~\cite{centola2018experimental_tipping,ye2021nat}, or by making conservative assumptions on the model parameters. 

In the following, we show that Problem~\ref{pr:1} under Assumption~\ref{a:delta} can be solved using either a conformity gain controller or an innovation gain controller, under some general conditions on the shape of the adaptation rate. Clearly, in this setting, the (dynamic) incentives should favor action $2$ by increasing the entries in the second row of the payoff matrix. Hence, the conformity gain controller uses matrix $\mat{G^{(4)}}$ from \eqref{eq:matrices}, and the innovation gain controller uses 
 $\mat{ G^{(3)}}$. The opposite problem (i.e., enforcing convergence to $x=1$) can be solved in a symmetric fashion, as discussed at the end of this section. 

\subsection{Conformity gain controller}

We consider the conformity gain controller with ${\mat G}=\mat{ G^{(4)}}$. We demonstrate that, under some reasonable assumptions on the adaptation rate $\phi$, the controller is able to steer the controlled system to the desired equilibrium

\begin{theorem}\label{th:approach_coordination}\rm
An adaptive-gain control $(\mat{G^{(4)}},\phi)$ solves Problem~\ref{pr:1} under Assumption~\ref{a:delta} with desired consensus equilibrium $\bar x=0$ for any initial condition  $x(0)\in[0,1)$ if the adaptive function $\phi$ is such that:
\begin{enumerate}
    \item $\phi(x)<0$ for $x\in[0,\delta)$;
    \item $\phi(x)>0$ for $x\in(\delta,1]$; and
    \item there exists a constant $\varepsilon^*>0$ such that $\phi(x)>a-c$ for all $x\in[1-\varepsilon^*,1]$.
\end{enumerate}
\end{theorem}
\begin{proof}
For the sake of readability, we define two positive constants $\alpha=a-c>0$ and $\beta=d-b>0$. Using this notation, 
and adopting ${\mat G}=\mat{ G^{(4)}}$ from \eqref{eq:matrices}, the controlled replicator equation in \eqref{eq:controlled_replicator} reduces to the following:
\begin{equation}\label{eq:controlled_approach_coordination}
    \begin{alignedat}{1}
            \dot x&=x\big(1-x\big)\big((\alpha+\beta)x-\beta-g(1-x)\big)\\
            \dot g&=g\phi(x).
    \end{alignedat}
\end{equation}
First, we verify that  \eqref{eq:controlled_approach_coordination} has the following equilibria:
    \begin{enumerate}
   \item $(x,g)=(0,0)$, which is (locally) exponentially stable;
       \item $(x,g)=(1,0)$, which is an (unstable) saddle point;
        \item $(x,g)=(\frac{\beta}{\alpha+\beta},0)$, which is an (unstable) source.
    \end{enumerate}
  The stability properties of the equilibria can be established from computing the Jacobian of \eqref{eq:controlled_approach_coordination}, which is equal to
  \begin{equation*}
   \begin{bmatrix}
-3x^2(\alpha+\beta+g)\hspace{-.05cm}+\hspace{-.05cm}2x(\alpha+2\beta+g)\hspace{-.05cm}-\hspace{-.05cm}\beta\hspace{-.05cm}-\hspace{-.05cm}g&-x(1-x)^2\\
     g\phi'(x)&\phi(x)
    \end{bmatrix}\hspace{-.1cm}.
\end{equation*}
In fact, at $(0,0)$, we get two negative eigenvalues $-\beta$ and $\phi(0)<0$; at $(1,0)$ we get the eigenvalues $-\alpha-\beta$ and $\phi(1)>0$; at $(\frac{\beta}{\alpha+\beta},0)$, we get the eigenvalues  $\frac{\alpha\beta}{\alpha+\beta}>0$ and $\phi(\frac{\beta}{\alpha+\beta})>0$, where the latter holds true because $\frac{\beta}{\alpha+\beta}$ is the unstable equilibrium of the uncontrolled dynamics. 
Assumption~\ref{a:delta} implies that $\frac{\beta}{\alpha+\beta}>\delta$ and, due to condition~2 of the theorem, $\phi(\frac{\beta}{\alpha+\beta})>0$.

We now analyze \eqref{eq:controlled_approach_coordination}, to prove convergence to $(0,0)$ from any initial condition with $g(0)>0$ and $x(0)<1$.  We divide the positively invariant domain $\Omega$ (see Lemma~\ref{lemma:invariance}) into three regions:
\begin{equation*}
    \mathcal A=\Big[0,\frac{\beta}{\alpha+\beta}\Big)\times\mathbb R_{+}, \,\mathcal B=\Big[\frac{\beta}{\alpha+\beta},1\Big)\times\mathbb R_{+}, \,\mathcal C=\{1\}\times\mathbb R_{+}.
\end{equation*}
Evidently, $\Omega=\mathcal A\cup\mathcal B\cup\mathcal C$. In the following, we prove that any trajectory with initial conditions in $\mathcal A\cup \mathcal B$ converges to $(0,0)$. The proof is structured in three steps: 1) we prove that any trajectory with initial conditions in $\mathcal A$ converges to $(0,0)$; 2) we rule out the existence of trajectories with $\omega$-set in the region $\mathcal B$; and 3) we prove that no trajectory starting from $\mathcal B$ can converge to $\mathcal C$, which yields the conclusion that all trajectories starting from $\mathcal A\cup \mathcal B$ necessarily converge to $(0,0)$.

In the first step, we observe that the first equation of the vector field in \eqref{eq:controlled_approach_coordination} is uniformly bounded for $(x,g)\in \mathcal A$ by
\begin{equation}
            \dot x\leq x(1-x)((\alpha+\beta)x-\beta),
\end{equation}
which is independent of $g$ and its right-hand side is strictly negative for $(x,g)\in \mathcal{A}$, guaranteeing that $x(t)\to 0$ for any initial condition in $\mathcal A$. As a consequence, there exists a time $\bar t\geq 0$ such that $x(t)<\delta$ for any $t\geq \bar t$. Due to condition~1) of the theorem, this implies that $\dot g$ is eventually strictly negative, and thus $g(t)\to 0$. Hence, the basin of attraction of $(0,0)$ includes $\mathcal A$. This completes the first step of the proof.

In the second step, we focus on the behavior of \eqref{eq:controlled_approach_coordination} in $\mathcal B$. First, define $\mu:=\min_{x\in[\frac{\beta}{\alpha+\beta},1]}\phi(x)$. As observed in the above, $x^*=\frac{\beta}{\alpha+\beta}$ is the unstable equilibrium of the uncontrolled dynamics. 
Assumption~\ref{a:delta} implies that $\frac{\beta}{\alpha+\beta}>\delta$ and, consequently, $\mu>0$. This observation allows us to bound the second equation in \eqref{eq:controlled_approach_coordination} uniformly with respect to $x$ in $\mathcal B$ as $\dot g\geq \mu g$. Using Gronwall's inequality~\cite{Pachpatte_book}, we conclude that $g(t)\geq g(0)\exp\{\mu t\}$, for any $t\geq 0$ provided that $x(s)\geq\frac{\beta}{\alpha+\beta}$ for all $s\in[0,t]$. An important consequence is that for any constant $\varepsilon>0$, we can bound
\begin{equation}\label{eq:bound}
    \dot x=x(1-x)((\alpha+\beta)x-\beta-g(1-x))\hspace{-.05cm}\leq \hspace{-.05cm}x(1-x)(\alpha-\varepsilon g(0)e^{\mu t}),
\end{equation}
for any $x\in[\frac{\beta}{\alpha+\beta},1-\varepsilon]$. Observe that the right side of \eqref{eq:bound}, and by implication $\dot x$, is strictly negative if $\alpha < \varepsilon g(0) e^{\mu t}$. Evidently, $\dot x < 0$ for all 
\begin{equation}
    t\geq\tau_\varepsilon:=\frac1\mu\ln(\alpha)-\frac1\mu\ln(\varepsilon)-\frac1\mu\ln(g(0)).
\end{equation} 

Based on these observations, we complete the second step of the proof by ruling out trajectories that have $\omega$-set in the region $\mathcal B$. To do so, we first show that $x(t)$ either goes below $\frac{\beta}{\alpha+\beta}$ (thus reaching $\mathcal A$, where we know that the desired equilibrium is attractive), or it converges to $1$ (reaching $\mathcal C$). The latter scenario will be ruled out in the third step. 

Let us consider a generic trajectory starting in $\mathcal B$ that does not converge to $\mathcal C$, i.e., does not converge to $x=1$. This implies the existence of a constant $\varepsilon>0$ such that for any $\tau>0$, there exists $t_1\geq\tau$ with $x(t_1)< 1-\varepsilon$. If we let $\tau=\tau_\varepsilon$ computed above, we conclude that at time $t_1$ we have $\dot x(t_1)<0$. For an infinitesimally small time increment $\Delta T>0$, we have $x(t_1+\Delta T) < x(t_1)$ and so there holds $\dot x(t_1+\Delta T)<0$, being that $x(t_1+\Delta T)\in[\frac{\beta}{\alpha+\beta},1-\varepsilon]$ and $t_1+\Delta T>\tau_\varepsilon$. By continuity, it follows that the trajectory is such that $x(t)$ necessarily decreases below $\frac{\beta}{\alpha+\beta}$, thereby reaching $\mathcal A$. This completes the second step of the proof.

In the third step, we rule out convergence to $\mathcal C$. Let us assume that there exists a trajectory of \eqref{eq:controlled_approach_coordination} with initial condition in $\mathcal B$ that converges to $x=1$. Hence, for any constant $\varepsilon> 0$ there exists $\bar t$ such that $x(t)> 1-\varepsilon$, for any $t>\bar t_\varepsilon$. Let us consider $\varepsilon=\varepsilon^*$ from the third condition of the theorem on $\phi(x)$. This implies that, for any $t>\bar t_\varepsilon$ it holds $\phi(x)>\alpha$. Moreover, since $\phi(x)>\alpha$ for all $x$ in the compact domain $[1-\varepsilon^*,1]$, we can refine the bound as $\phi(x)\geq\alpha+\sigma$, for some constant $\sigma>0$. Hence, similar to the above, we can bound the second equation of \eqref{eq:controlled_approach_coordination} as $\dot g\geq (\alpha+\sigma) g$ for any $t\geq \bar t_\varepsilon$. As a consequence, for $t\geq\bar t_\varepsilon$, there holds 
\begin{equation}\label{eq:bound1}
    g(t)\geq g(\bar t)e^{(\alpha+\sigma)(t-\bar t_\varepsilon)}\geq g(0)e^{(\alpha+\sigma)(t-\bar t_\varepsilon)},
\end{equation}
being $\dot g>0$ in $\mathcal B\cup\mathcal C$. On the other hand, if we define $y = 1-x$, and notice that $\dot y = -\dot x$, it follows from \eqref{eq:controlled_approach_coordination} that
\begin{equation}\begin{alignedat}{1}
    \dot y   &= -(1-y)y(\alpha-(\alpha+\beta)y-gy)\\&= -\alpha y + y^2(2\alpha+\beta+g) - y^3 (\alpha+\beta+g)\geq -\alpha y. 
\end{alignedat}\end{equation}
Hence Gronwall's inequality yields that 
\begin{equation}\label{eq:bound2}
    1-x(t)\geq (1-x(0))e^{-\alpha t}.
\end{equation} To conclude, by inserting the two bounds in \eqref{eq:bound1} and \eqref{eq:bound2} into the first equation of \eqref{eq:controlled_approach_coordination}, we obtain
\begin{equation}
    \dot x  \leq  (1-x)x(\alpha-(1-x(0))g(0)e^{-(\alpha+\sigma)\bar t_\varepsilon}e^{\sigma t}).
\end{equation}
Using an argument similar to the one used in \eqref{eq:bound}, the fact that the last term increases exponentially fast guarantees that there exists a time $\tau_1$ such that $\dot x(\tau_1) < 0$ and, by continuity, $\dot x(t) < 0$ for all $t\geq \tau_1$ such that $x(t)\geq \delta$; ultimately, there will be a time $\tau_2 > \tau_1$ such that $x(\tau_2) < 1-\varepsilon$. This contradicts the fact that $x(t)\to 1$, yielding the claim that no trajectory can converge to $\mathcal C$, thereby completing the proof.
\end{proof}
 
Theorem~\ref{th:approach_coordination} poses only a few mild constraints on the shape of the adaptation rate $\phi$, as we now discuss.
   
\begin{remark}\label{rem:adaptation}
 The first two constraints on $\phi$ in Theorem~\ref{th:approach_coordination} require only local information on the desired equilibrium ---see Assumption~\ref{a:delta} and the related discussion. Several classes of functions can be defined to meet these two conditions, e.g., 
    \begin{equation}\label{eq:affine}
        \phi(x)=p(x^q-\delta^q),
    \end{equation}
    with positive parameters $p>0$ and $q>0$ that regulate the velocity of the adaptation process and its reactivity, respectively. This class of functions encompasses the affine functions used in~\cite{zino2023_adaptive} (case $q=1$). More generally, one can use compositions of \eqref{eq:affine} with functions that preserve the sign.
    \end{remark}
\begin{remark}\label{rem:adaptation2}
Condition 3) in Theorem~\ref{th:approach_coordination} requires some information on the order of magnitude of the payoff function. However, it is worth noticing that one does not need full knowledge of the payoff matrix in order to guarantee that Condition 3) is satisfied. In fact, one only needs an estimate of the difference between the entries $a-c$ or, conservatively, just an upper bound on the maximum component of the payoff matrix, which can be ultimately estimated online, from observations of the velocity of the dynamics. For instance, when using \eqref{eq:affine}, it suffices to guarantee that $p$ is  larger than (an estimate of) the maximum entry of  $\mat A$.
\end{remark}




\subsection{Innovation gain controller}

Now, we consider the innovation gain controller (${\mat G}=\mat{ G^{(3)}}$), which  represents providing an additional (temporary and adaptive) payoff to early adopters. In the following, we prove that this second controller is also able to solve Problem~\ref{pr:1}, as demonstrated in the following result.

\begin{theorem}\label{th:approach_innovation}\rm
The adaptive-gain control $(\mat{G^{(3)}},\phi)$ solves Problem~\ref{pr:1} under Assumption~\ref{a:delta} with desired consensus equilibrium $\bar x=0$ for any initial condition  $x(0)\in[0,1)$ if the adaptive function $\phi$ is such that:
\begin{enumerate}
    \item $\phi(x)<0$ for $x\in[0,\delta)$; and
    \item $\phi(x)>0$ for $x\in(\delta,1]$.
\end{enumerate}
\end{theorem}
\begin{proof}
Let $\alpha=a-c>0$ and $\beta=d-b>0$. Using~\eqref{eq:matrices}, \eqref{eq:controlled_replicator} reduces to the planar system
\begin{equation}\label{eq:controlled_approach_innovation}
    \begin{alignedat}{1}
            \dot x&=x(1-x)((\alpha+\beta)x-\beta-gx)\\
            \dot g&=g\phi(x),
    \end{alignedat}
\end{equation}
which has the following equilibria:
    \begin{enumerate}
        \item $(x,g)=(0,0)$, which is (locally) asymptotically stable;
        \item $(x,g)=(1,0)$, which is an (unstable) saddle point;
        \item $(x,g)=(\frac{\beta}{\alpha+\beta},0)$, which is an (unstable) source,
    \end{enumerate}
    where the stability is determined from the Jacobian matrix.
    
The convergence analysis follows the proof of Theorem~\ref{th:approach_coordination}, but it is slightly simpler. We thus omit providing the full details but instead focus on the key differences. First, we divide the domain $\Omega$ of \eqref{eq:controlled_approach_innovation} into $\mathcal A=[0,\frac{\beta}{\alpha+\beta})\times\mathbb R_{+}$, $\mathcal B=[\frac{\beta}{\alpha+\beta},1)\times\mathbb R_{+}$, and $\mathcal C=\{1\}\times\mathbb R_{+}$, with $\Omega=\mathcal A\cup\mathcal B\cup\mathcal C$. In $\mathcal A$, the same argument used in Theorem~\ref{th:approach_coordination} guarantees convergence of \eqref{eq:controlled_approach_innovation} to $(0,0)$. In $\mathcal B$, we use Gronwall's inequality to derive a bound on $g(t)$ of the form $g(t)\geq g(0)e^{\mu t}$, with $\mu=\min_{x\in[\frac{\beta}{\alpha+\beta},1]}\phi(x)>0$. However, the analysis of the system in $\mathcal B$ is simpler, since we do not need to rule out convergence to $\mathcal C$. In fact, the bound $g(t)\geq g(0)e^{\mu t}$ can now be directly inserted into the first equation of \eqref{eq:controlled_approach_innovation}, yielding
\begin{equation}\begin{alignedat}{1}
        \dot x(t)&\leq x(t)(1-x(t))\left((\alpha+\beta)x(t)\hspace{-.05cm}-\hspace{-.05cm}\beta\hspace{-.05cm}-\hspace{-.05cm}x(t)g(0)e^{\mu t}\right)\\&\leq x(t)(1-x(t))\Big(\alpha-\frac{\beta}{\alpha+\beta}g(0)e^{\mu t}\Big)
\end{alignedat}
\end{equation}
which holds for any $x\in[\frac{\beta}{\alpha+\beta},1)$ and is strictly negative for \begin{equation}
    t\geq\bar t=\frac{1}{\mu}\ln(\alpha)+\frac{1}{\mu}\ln(\alpha+\beta)-\frac{1}{\mu}\ln(\beta)-\frac{1}{\mu}\ln (g(0)).
\end{equation}
This guarantees that $x(t)$ is monotonically decreasing for $t\geq \bar t$, until it reaches $\mathcal A$, yielding the claim.
\end{proof}

Theorem~\ref{th:approach_innovation} guarantees that one can solve Problem~\ref{pr:1} using the  innovation gain controller, with an adaptation rate that satisfies the mild constraints discussed in Remark~\ref{rem:adaptation}. However, in this case, no additional constraints regarding the magnitude of the adaptation rate close to $x=1$ are needed, unlike the conformity gain controller, where the additional condition discussed in Remark~\ref{rem:adaptation2} is needed. 


\subsection{Comparison between the two control approaches}


\begin{figure}
    \centering
\subfloat[Conformity gain]{\includegraphics[]{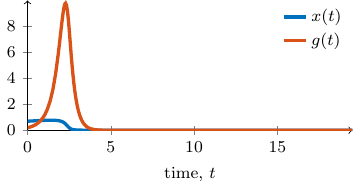}\label{fig:3a}}\quad\subfloat[Innovation gain]{\includegraphics[]{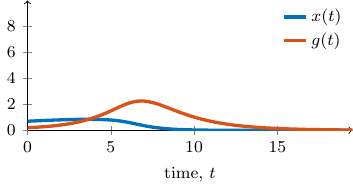}\label{fig:3b}}
    \caption{Trajectories of the controlled pure coordination game with $a=d=1$ and $x(0)=0.7$ using (a) conformity gain control ($q=1$ and $p=7$) and (b) innovation gain control ($q=p=1$). \label{fig:3}}
\end{figure}

We conclude this section by comparing the two approaches via numerical simulations. Figure~\ref{fig:3} shows two exemplifying trajectories of the controlled replicator equation. In both cases, we consider a pure coordination game (Example~\ref{ex:pure}) with $a=d=1$ (and $b=c=0$) and initial condition $x(0)=0.7$. Being $x(0)>x^*=1/2$, the uncontrolled system converges to  $1$. For both scenarios, we use an affine adaptive function (with $q=1$ in \eqref{eq:affine}), which always satisfies conditions 1) and 2) of Theorems~\ref{th:approach_coordination} and~\ref{th:approach_innovation}, and we set $\delta=0.4$. In Fig.~\ref{fig:3a}, we use conformity gain ($p=7$, which satisfies condition 3) of Theorem~\ref{th:approach_coordination}); in Fig.~\ref{fig:3b}, we use innovation gain ($p=1$). 

From the simulations in Fig.~\ref{fig:3}, we observe that both our controllers successfully steer the system to the desired equilibrium $\bar x=0$. These simulations give us also the opportunity to discuss some of the pros and cons of the two  approaches. Specifically, we observe that the conformity gain control appears to yield faster convergence. This is due to the fact that the gain is multiplied by $1-x(t)$ instead of by $x(t)$ in the $\dot x$ equation. Therefore, its effect does not vanish as $x(t)$ approaches $0$. The drawback, however, is that larger peaks in the gain are reached when the system is far from the desired equilibrium, as can be observed in Fig.~\ref{fig:3a}. 

In view of these observations, a  balanced control strategy can be designed by combining the two approaches, whereby innovation gain ($\mat{G^{(3)}}$) is used in the first stage (i.e., when the system is far from the desired equilibrium), and then the controller switches to a conformity gain ($\mat{G^{(4)}}$). From a practical point of view, such a combined approach may entail a first stage in which policies should incentivize innovators and early adopters, followed by a second stage where instead the collective adoption of the desired action is promoted.

Besides integrating the two different approaches, the constraints in our theoretical results leave many other degrees of freedom in designing the adaptation rate $\phi$. Hence, the optimal design of such a function is of paramount importance toward improving the performance of our controller. To investigate this issue, we perform a set of simulations, in which we numerically integrate \eqref{eq:controlled_replicator} over a time horizon $T$, using the innovation gain control and different adaptation rates in the form of \eqref{eq:affine}, exploring how the parameters $p$ and $q$ affect the total control effort $J_g:=\int_0^{T} g(t)dt$ and the maximum gain  $g_{\text{max}}:=\max_{t\in[0,T]} g(t)$. Our results, reported in Fig.~\ref{fig:4}, depicts an interesting scenario. In fact, while the total control effort displays a monotonic behavior, whereby it grows as $q$ increase and $p$ decreases, the peak gain shows a more complex, non-monotonic behavior. Specifically, $g_{\text{max}}$ seems to be particularly large when $p$ is large and $q$ has moderate values, while its minimum is achieved in an intermediate region for the parameter $p$, which depends on $q$. Apparently, an optimal trade-off can be obtained by setting moderate levels of $p$ and small values of $q$, similar to \eqref{fig:3b}. Alternatively, for larger values of $q$, the plot in Fig.~\ref{fig:4b} seems to describe a linear relation between the parameters $q$ and $p$ that minimizes the maximum peak gain, whereby the optimal value of $p$ (initially small) increases linearly with $q$.

\begin{figure}
    \centering
\subfloat[Total control effort, $J_{g}$]{\includegraphics[]{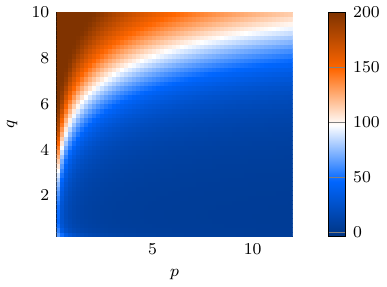}\label{fig:4a}}\quad\subfloat[Peak gain, $g_{\text{max}}$]{\includegraphics[]{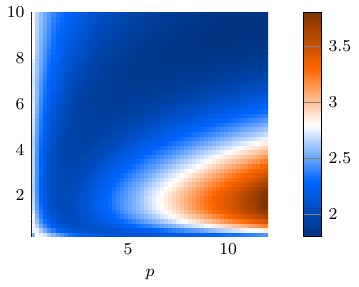}\label{fig:4b}}
    \caption{Total control effort and peak gain of the innovation gain control for a pure coordination game ($a=d=1$) using adaptation rate \eqref{eq:affine} with different values of $p$ and $q$. \label{fig:4}}
\end{figure}

\begin{remark}\label{rem:other_eq1}
    In this section, we focused on the scenario in which the desired consensus is $\bar x=0$. Similarly, an adaptive-gain controller can be designed to steer the system to the other locally stable equilibrium of a coordination game, i.e., $\bar x=1$. Specifically, one can use an innovation gain controller with $\mat{ G^{(2)}}$, imposing $\phi(x)>0$ for $x\in[0,1-\delta)$ and $\phi(x)<0$ for $x\in(1-\delta,1]$; or a conformity gain controller with $\mat{ G^{(1)}}$ and imposing the additional condition that there exists a constant $\varepsilon^*>0$ such that $\phi(x)>a-c$ for all $x\in[0,\varepsilon^*]$.
\end{remark}

\section{Consensus Stabilization Problem}\label{sec:dominant}

The consensus stabilization problem described in Problem~\ref{pr:2} arises whenever the desired consensus equilibrium is unstable for the uncontrolled dynamics. Without loss of generality, we focus on the case in which the consensus equilibrium $\bar x=0$ is unstable, and our objective is to steer the system to it. The opposite scenario, in which the desired equilibrium is $\bar x=1$ and it is unstable, can be treated similarly (see a remark at the end of this section).

\begin{assumption}\label{a:consensus}
    The uncontrolled replicator equation \eqref{eq:replicator} admits the unstable equilibrium $\bar x=0$. 
\end{assumption}

We start by presenting an impossibility result for the innovation gain controller (control matrix $\mat{G^{(3)}}$). Then, we focus on the conformity gain controller ($\mat{G^{(4)}}$), proving that such a controller is  able to solve Problem~\ref{pr:2}. 

\begin{proposition}\label{prop:impossible}\rm
    No innovation gain controller $(\mat{G^{(3)}},\phi)$ can solve  Problem~\ref{pr:2} under Assumption~\ref{a:consensus} with desired consensus equilibrium $\bar x=0$.
\end{proposition}
\begin{proof}
    We prove the statement by way of contradiction. Assume that there exists an adaptation rate $\phi(x)$ such that the innovation gain controller solves Problem~\ref{pr:2}. This implies that $\lim_{t\to\infty} g(t)=\bar g$ for some constant $\bar g$, and thus for any constant $\varepsilon>0$ there exists $\bar t>0$ such that $g(t)\leq \bar g+\varepsilon$ for all $t\geq \bar t$. Therefore, inserting $\mat{G^{(3)}}$ into \eqref{eq:controlled_replicator}, we bound \begin{equation}\label{eq:bound_impossible}
        \begin{alignedat}{1}\dot x&= x(1-x)((a+d-b-c)x+b-d-gx)\\&\geq x(1-x)((a+d-b-c)x+b-d-(\bar g+\varepsilon)x),\end{alignedat}
    \end{equation}
    which is a one-dimensional equation, independent of the dynamics of $g(t)$. 
    
Let us consider the equation
\begin{equation}\label{eq:bound_y}
    \dot y=x(1-y)((a+d-b-c)y+b-d-(\bar g+\varepsilon)y),
\end{equation}
whose right-hand side is the bound from \eqref{eq:bound_impossible}. For \eqref{eq:bound_y}, $\bar y=0$ has the same linear stability properties of $\bar x=0$ for the uncontrolled system in \eqref{eq:replicator}, determined by linearization. Therefore, it is unstable and it cannot be attractive. Finally, Gronwall's inequality on \eqref{eq:bound_impossible}, we conclude that $x(t)\geq y(t)$, solution of \eqref{eq:bound_y} with initial condition $y(0)=x(0)$. Since $y(t)$ does not converge to $0$, necessarily $x(t)$ cannot converge to $0$ too, yielding the claim.
\end{proof}

\begin{theorem}\label{th:stabilization}\rm
The adaptive-gain control $(\mat{G^{(4)}},\phi)$ solves Problem~\ref{pr:2} under Assumption~\ref{a:consensus} with desired consensus equilibrium $\bar x=0$ for any initial condition  $x(0)\in[0,1)$ if the adaptive function $\phi$ is such that:
\begin{enumerate}
    \item $\phi(x)>0$ for $x\in(0,1]$; and
    \item there exist $h,k>0$ such that $\lim_{x\to 0^+}\phi(x)/x^h=k$. 
    \end{enumerate}
Moreover, if the game is a dominant-strategy game ($a>c$ and $b>d$), we further require that
\begin{enumerate}\setcounter{enumi}{2}
    \item exists $\varepsilon^*>0$ such that $\phi(x)>a-c$ for $x\in[1-\varepsilon^*,1]$.
    \end{enumerate}
\end{theorem}
\begin{proof}
We start by considering the scenario of a dominant-strategy game in which coordination on action $1$ is the NE ($a>c$ and $b>d$). Similar to the proof of Theorem~\ref{th:approach_coordination}, we define two strictly positive constants $\alpha=a-c>0$ and $\beta=b-d>0$, noting $\beta$ differs from the proof of Theorem~\ref{th:approach_coordination} by the sign. In such a scenario, the controlled replicator dynamics in \eqref{eq:controlled_replicator} reduces to the following planar system of ODEs:
\begin{equation}\label{eq:controlled_dominant_innovation}
    \begin{alignedat}{1}
            \dot x&=x(1-x)(\alpha x +(\beta-g)(1-x))\\
            \dot g&=g\phi(x).
    \end{alignedat}
\end{equation}

Being $\phi(x)\geq 0$, then $\dot g\geq 0$, and thus $g(t)$ is monotonically non-decreasing. So either i) $g(t)\to\bar g$ or ii) $g(t)\to\infty$. In scenario i), since $\lim_{t\to\infty}g(t)=\bar g$, then necessarily $\lim_{t\to\infty}\dot g(t)= 0$. This implies $\lim_{t\to\infty}\phi(x(t))=0$ and, ultimately, $\lim_{t\to\infty}x(t)=0$ due to assumption 1) in the Theorem's statement, solving Problem~\ref{pr:2}. In the following, we rule out scenario ii) by way of contradiction.

To obtain a contradiction, assume that $g(t)\to\infty$. First, we exclude the possibility that $x(t)$ converges to $1$. This is done by following the same line of reasoning used in the proof of Theorem~\ref{th:approach_coordination}. In fact, due to condition 3) of the Theorem's statement, $g(t)$ grows in a neighborhood of $x=1$ faster than an exponentially-growing function with exponent $\alpha$. Hence, $g(t)>\beta$ in finite time $\tau_1$, yielding $\dot x\leq \alpha x(1-x)$ for any $t\geq \tau_1$. Consequently, $x(t)\to 1$ at an exponential rate no larger than $\alpha$. Combining this with the fact that $g(t)$ grows exponentially at a rate larger than $\alpha$, it follows that the term $g(1-x)$ grows exponentially large, guaranteeing that there exists a time $\tau_2$ such that $\dot x(t)<0$ for all $t\geq \tau_2$, and thus $x(t)$ eventually decreases and cannot converge to $1$.

Since $x(t)$ does not converge to $1$ and we have assumed that $g(t)\to\infty$, then there exist two positive constants $\varepsilon_1>0$ and $\tau_3<\infty$ such that i) $x(\tau_3)\leq 1-\varepsilon_1$ and ii) $g(\tau_3)\geq \alpha/\varepsilon_1+\beta+\gamma$, where $\gamma>0$ is a sufficiently small but positive constant. This implies that $\dot x(\tau_3)\leq -\gamma x(1-x)^2$. Using again the continuity argument used in the proof of Theorem~\ref{th:approach_coordination} and the fact that $g(t)$ is monotonically increasing, we conclude that the bound $\dot x(t)\leq -\gamma \varepsilon_1^2 x$ holds true for any $t\geq\tau_3$. Hence, we  use Gronwall's lemma to derive an exponential bound on the convergence of $x(t)$ to $0$ as 
\begin{equation}\label{eq:exp_conv_bound}
    x(t)\leq (1-\varepsilon_1)e^{-\gamma\varepsilon_1^2(t-\tau_3)},
\end{equation} for any $t\geq \tau_3$. 

At this stage, due to the definition of limit and using condition 2) of the Theorem's statement, we establish that there necessarily exists a constant $\varepsilon_2>0$ such that $\phi(x)\leq 2kx^h$. Moreover, using the exponential convergence bound on $x(t)$ established in \eqref{eq:exp_conv_bound}, we  guarantee that $x(t)<\varepsilon_2$, for any time $t>\tau_4:=\max\{\tau_3,\tau_3-\gamma-\varepsilon_1^2-\ln(\varepsilon_2)+\ln(1-\varepsilon_1)\}$.

Observe that $M:=\max_{x\in[0,1]}\phi(x)<\infty$, being the maximum of a continuous function over a compact set. Hence, we are ready to obtain the last piece of our proof by contradiction, by integrating the second equation in \eqref{eq:controlled_dominant_innovation}, obtaining         $g(t)=g(0)e^{k\int_0^t \phi(x(s))ds}$. Finally, by using that $\phi(x(t))\leq M$ for all $t\leq\tau_4$ and \begin{equation}
    \phi(x(t))\leq 2kx^h\leq 2k(1-\varepsilon_1)^he^{-h\gamma\varepsilon_1^2(t-\tau_3)}
\end{equation}
for all $t>\tau_4$, we obtain the following bound:
    \begin{equation}\begin{alignedat}{1}
       \displaystyle \lim_{t\to\infty}g(t) &=   \displaystyle\lim_{t\to\infty}g(0)e^{\int_0^t \phi(x(s))ds} \\
       &=     g(0)e^{\int_0^{\tau_4}\phi(x(s)) ds}e^{\int_{\tau_4}^{\infty}\phi(x(s))ds}\\
       &\leq g(0)e^{\int_0^{\tau_4}M ds}e^{k(1-\varepsilon_1)^h\int_{\tau_4}^{\infty}e^{-h\gamma\varepsilon_1^2(t-\tau_4)} ds}\\
       &\leq g(0)e^{M\tau_4}e^{k(1-\varepsilon)^h\int_0^\infty e^{-h\gamma\varepsilon_1^2s}ds}\\
     &\leq  g(0)e^{M\tau_4}e^{\frac{k(1-\varepsilon_1)^h}{h\gamma\varepsilon_1^2}}<+\infty,
    \end{alignedat}\end{equation}
    which contradicts the assumption that $g(t)\to\infty$, yielding the claim.

Now, we consider the scenario of an anti-coordination game. Defining two strictly positive constants $\alpha=c-a>0$ and $\beta=b-d>0$, the controlled system reduces to
\begin{equation}\label{eq:controlled_anticoordination_innovation}
    \begin{alignedat}{1}
            \dot x&=x\big(1-x\big)\big(-\alpha x+(\beta-g) (1-x)\big)\\
            \dot g&=g\phi(x).
    \end{alignedat}
\end{equation}
The analysis is similar to the one of the dominant-strategy game in \eqref{eq:controlled_dominant_innovation}. Specifically, we proceed by contradiction, showing that the limit of $g(t)$ exists finite and thus $x(t)$ must converge to $0$. It worth noticing that the negative sign in the first equation of \eqref{eq:controlled_anticoordination_innovation} simplifies the treatment. In fact, $\dot x$ is always negative in a neighborhood of $x=1$, regardless of $g(t)$. This implies that convergence to $x=1$ is excluded without needing any additional condition on $\phi$. The rest of the proof then follows the identical line of arguments used for the dominant-strategy game, and detailed calculations are omitted. \end{proof}

Different from the consensus reaching problem, the consensus stabilization problem can only be solved using a conformity gain controller, i.e., by acting on the diagonal entry of the payoff matrix. However, assuming that the diagonal matrix is controllable, then the conditions required for the adaptive function $\phi$ are quite mild, as discussed in the following.

\begin{remark}
     The first two conditions on $\phi$ in Theorem~\ref{th:stabilization} can be easily satisfied by any polynomial function that is strictly positive in $(0,1]$ and has no constant term. For instance, we can define the class of adaptation rate:
     \begin{equation}\label{eq:power}
         \phi(x)=px^q,
     \end{equation}
     with parameters $p>0$ and $q>0$ regulating velocity and reactivity of the adaptation process, respectively, which encompasses the linear case ($q=1$) considered in~\cite{zino2023_adaptive}. The third condition, which is needed only for dominant-strategy games, is similar to the one discussed in Remark~\ref{rem:adaptation2} and  requires limited information on the structure of the payoff function. 
\end{remark}
\begin{remark}\label{rem:other_eq2}
    In this section, we focused on the scenario in which the goal is to steer the system to the unstable equilibrium $\bar x=0$. Similarly, we can design an adaptive-gain controller to steer the system to $\bar x=1$ when unstable by utilizing   a conformity gain controller with $\mat{G^{(1)}}$, replacing condition 2) in Theorem~\ref{th:stabilization} with requiring the existence of two positive constants $h,k>0$ such that $\lim_{x\to 1^-}\phi(x)/(1-x)^h=k$, and condition 3), when needed, with the one already presented in Remark~\ref{rem:other_eq1}.
\end{remark}

\begin{figure}
    \centering
\subfloat[Trajectory]{\includegraphics[]{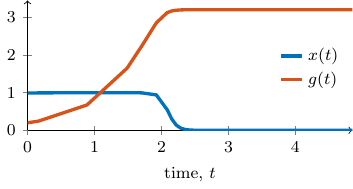}\label{fig:5a}}\quad\subfloat[Final gain, $\bar g$]{\includegraphics[]{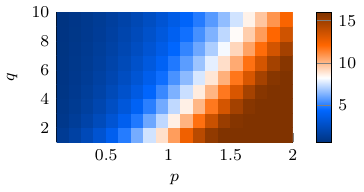}\label{fig:5b}}
    \caption{Controlled prisoner's dilemma with $a=1$, $b=3$, $c=0$, $d=2$ and $x(0)=0.99$, using conformity gain control with $\phi$ from \eqref{eq:power}. In (a), we report a trajectory of the controlled dynamics (with $p=0.4$ and $q=1$); in (b) the final gain $\bar g$, for different choices of the parameters $p$ and $q$. \label{fig:5}}
\end{figure}

Finally, it is worth noticing that the adaptation rate used in Theorem~\ref{th:stabilization} yields a monotonically increasing gain (see Fig.~\ref{fig:5a}). While sufficient to solve Problem~\ref{pr:2}, this might result in a much larger $\bar g$ than actually needed to guarantee convergence to the desired equilibrium. Indeed, the choice of the adaptation rate plays an important role in determining the asymptotic value of the gain, as can be observed in Fig.~\ref{fig:5b}). In particular, it seems that slow-to-moderate adaptation velocities (small $p$) with large reactivity $q$ yields better performance, i.e., smaller asymptotic gains. Moreover, another viable method may involve adding additional decay terms to the gain dynamics so $g(t)$ is not monotonic, similar to those proposed in~\cite{mei2016adaptive} for the consensus problem.

\section{Set-point regulation Problem}\label{sec:setpoint}

The set-point regulation problem, viz. Problem~\ref{pr:3}, consists of steering the system to a mixed strategy state $x=\bar x$, which in general needs not be an equilibrium of the uncontrolled dynamics. The more complex nature of this problem calls for the development of specific implementations of our adaptive-gain control method, tailored to the characteristics of the game in question. 

\subsection{Dominant-strategy games}\label{sec:dominant_sp}

We start by considering a dominant-strategy game. Without loss of generality, we focus on the case in which the dominant strategy is action~$1$ and so the uncontrolled dynamics would converge to $x=1$ from any initial condition $x(0)\in(0,1]$.

In order to design the adaptive-gain controller, we recall the behavior of the uncontrolled dynamics (studied in Proposition~\ref{prop:convergence_uncontrolled}). In particular, we recall that the uncontrolled dynamics has always a drift in the direction of increasing $x$, since the right-hand side of the ODE is always strictly positive in $(0,1)$, and the equilibrium $x=1$ is globally attractive from any initial condition in $(0,1)$. Hence, we make the following observations. First, the controller should (adaptively) increase the payoff associated with action~$2$. As a consequence, innovation-gain control will use matrix $\mat G^{(3)}$ and coordination-gain control will use $\mat G^{(4)}$. Second, it is necessary to control the system when $x(t)$ is larger than the desired equilibrium $\bar x$, while below $\bar x$ the uncontrolled dynamics is sufficient to drive the system to $\bar x$. This, will guide us in the design of the adaptation rate $\phi(x)$. Third, it is key for the controller to prevent convergence to $x=1$. To this aim, innovation-gain control ($\mat G^{(3)}$) seems more efficient as it requires no additional conditions (see Theorems~\ref{th:approach_coordination} and~\ref{th:approach_innovation}). Based on these considerations, we define the adaptation rate adopting the form of a proportional controller, i.e., 
\begin{equation}\label{eq:linear_phi}
    \phi(x)=p(x-\bar x),
\end{equation}
where $p>0$ captures the velocity of the adaptation process. In the following, we prove that this control action guarantees solution of Problem~\ref{pr:3}. 

\begin{figure*}
    \centering
\subfloat[]{\includegraphics[]{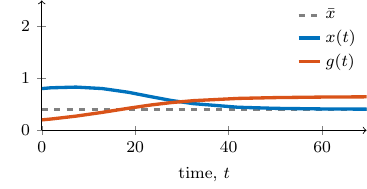}\label{fig:6a}}\subfloat[]{\includegraphics[]{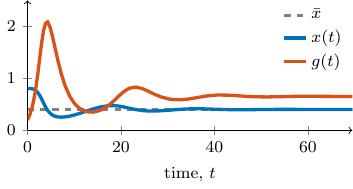}\label{fig:6b}}\\\subfloat[]{\includegraphics[]{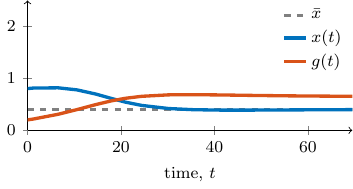}\label{fig:6c}}\subfloat[]{\includegraphics[]{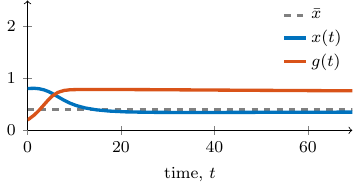}\label{fig:6d}}\\
\subfloat[]{\includegraphics[]{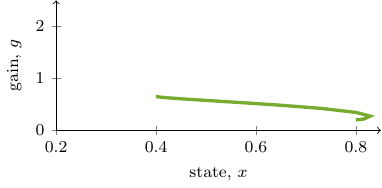}\label{fig:6e}}\subfloat[]{\includegraphics[]{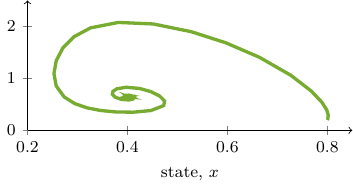}\label{fig:6f}}\\\subfloat[]{\includegraphics[]{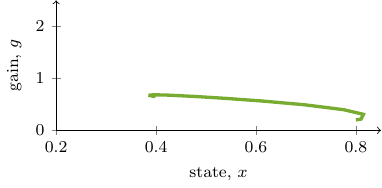}\label{fig:6g}}\subfloat[]{\includegraphics[]{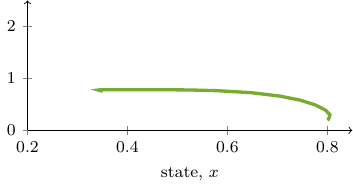}\label{fig:6h}}
    \caption{Simulations of the set-point regulation problem for a prisoners' dilemma (dominant-strategy game). Trajectories are represented in terms of (a--d) time-evolution and (e---h) state-space evolution. In (a,e) and (b,f), we consider the adaptive-gain controller from Theorem~\ref{th:creation} with $p=0.1$ and $p=2$, respectively. Panels (c,g) and (d,h) show controlled trajectories with nonlinear adaptation rates: with saturation in (c,g) and with fast-waning rate in (d,h). }
    \label{fig:case3}
\end{figure*}

\begin{theorem}\label{th:creation}\rm
The adaptive-gain control $(\mat{G^{(3)}},\phi)$ with $\phi$ defined in \eqref{eq:linear_phi} solves Problem~\ref{pr:3} for a dominant-strategy game with $a>c$ and $b>d$, 
for any initial condition  $x(0)\in(0,1)$. Precisely, the state converges to the desired equilibrium $\bar x$ and the gain converges to 
\begin{equation}
    \label{eq:g_bar4}
    \bar g=a-c+(b-d)\frac{1-\bar x}{\bar x}.
\end{equation}
\end{theorem}\medskip

\begin{proof}
Similar to the proof of Theorem~\ref{th:stabilization}, let $\alpha=a-c>0$ and $\beta=b-d>0$. With $\mat{G}=\mat{G^{(3)}}$, the controlled replicator equation reads
\begin{equation}\label{eq:controlled_dominant_creation}
    \begin{alignedat}{1}
            \dot x=&x\big(1-x\big)\big(\alpha x+\beta(1-x)-gx\big)\\
            \dot g=&gp(x-\bar x).
    \end{alignedat}
\end{equation}
We start by characterizing the equilibria of \eqref{eq:controlled_dominant_creation}, which are:
\begin{enumerate}
    \item $(x,g)=(0,0)$;
    \item $(x,g)=(1,0)$; and
    \item $(x,g)=(\bar x, \bar g)$, with $\bar g$ from \eqref{eq:g_bar4}.
\end{enumerate}

The Jacobian of \eqref{eq:controlled_dominant_creation} is equal to
 \begin{align*}
\begin{bmatrix}
   \alpha+\beta-2x(\alpha+2\beta+g)+3x^2(\beta+g)&-x^2(1-x)\\
     gp&p(x-\bar x)
    \end{bmatrix}.
\end{align*}
Evaluating it at the three equilibrium points, we observe that $(0,0)$ and $(1,0)$ are unstable saddle points, while $(\bar x, \bar g)$ is always (locally) exponentially stable. The last fact is obtained using the trace-determinant method: 
the determinant is always positive and the trace is negative.

We are now left to prove global convergence. It is helpful to re-write the first equation in \eqref{eq:controlled_dominant_creation} as
\begin{equation}\label{eq:controlled_dominant_creation2}
    \begin{alignedat}{1}
            \dot x&=x(1-x)(\alpha x+\beta(1-x)\hspace{-.05cm}-\hspace{-.05cm}\bar gx)\hspace{-.05cm}-\hspace{-.05cm}x^2(1\hspace{-.05cm}-\hspace{-.05cm}x)(g-\bar g)\\&=\displaystyle\beta {\bar x}^{-1}x(1-x)(\bar x-x)-x^2(1-x)(g-\bar g),
    \end{alignedat}
\end{equation}
from which we can write the entire system by separating the ODEs in the sum of two terms: $F(x,g)$ that depends on both variables, and $G(x)$ that depends only on $x$, obtaining
\begin{equation}\label{eq:controlled_dominant_creation_reduced}
           \begin{bmatrix}
               \dot x\\
               \dot g
           \end{bmatrix} =\begin{bmatrix}-x^2(1-x)(g-\bar g)\\gp(x-\bar x)\end{bmatrix}+\begin{bmatrix}\beta\bar x^{-1} x(1-x)(\bar x-x)\\0\end{bmatrix}.
\end{equation}

We start by focusing on the reduced system that considers only the first block, viz. $F(x,g)$. That is, 
\begin{equation}\label{eq:controlled_periodic}
          \begin{alignedat}{1}\displaystyle\dot x=&-x^2(1-x)(g-\bar g)\\\displaystyle\dot g=&gp(x-\bar x),\end{alignedat}
\end{equation}
from which we derive
\begin{equation}
    \frac{dx}{dg}=-\frac{x^2(1-x)(g-\bar g)}{pg(x-\bar x)}\implies  \frac{x-\bar x}{x^2(1-x)}dx=-\frac{(g-\bar g)}{pg}dg,
\end{equation}
which is a separable ODE that can be solved analytically by integrating both sides, obtaining the first integral of \eqref{eq:controlled_periodic}:
\begin{equation}\label{eq:lyapunov}
    V(x,g)=\frac{\bar x}{x}+(1-\bar x)\ln \Big(\frac{x}{1-x}\Big)+\frac{1}{p}\big(g-\bar g\ln g\big)+K,
\end{equation}
where $K\in R$ is a constant. In other words, we can easily check that $V(x,g)$ is conserved along trajectories of \eqref{eq:controlled_periodic} by explicitly computing
\begin{equation}\label{eq:nablaV}
    \nabla V(x,g)= \begin{bmatrix}
                \frac{\partial V}{\partial x}\\
                \frac{\partial V}{\partial g}
           \end{bmatrix}= \begin{bmatrix}
                \frac{x-\bar x}{x^2(1-x)}\\
                \frac{1}{p}\big(1-\frac{\bar g}{g}\big)
           \end{bmatrix},
\end{equation}
and computing the scalar product $\nabla V(x,g)^\top F(x,g)$. We further observe that $V$ has a global minimum at $(\bar x,\bar g)$. Hence, we can enforce $V(\bar x,\bar g)=0$ by setting 
\begin{equation}
    K=-1-(1-\bar x)\big(\ln\bar x-\ln(1-\bar x)\big)\frac{\bar g}{p}(1-\ln \bar g).
\end{equation}
Consequently, it holds true that $V(x,g)>0$ for all $x\in[0,1]$ and $g\geq 0$, and $V(x,g)=0$ if and only if $x=\bar x$ and $g=\bar g$. Finally, we observe that $V$ diverges to $+\infty$ at the boundaries of the domain $[0,1]\times[0,\infty]$. These properties make $V(x,g)$ a suitable Lyapunov candidate function.

Now, we verify that the Lyapunov candidate function $V(x,g)$ satisfies the property $\frac{dV}{dt}\leq 0$ along any trajectory of the controlled replicator equation in \eqref{eq:controlled_dominant_creation_reduced} and $\frac{dV}{dt}= 0$ if and only if $x=\bar x$. Using \eqref{eq:nablaV}, we compute
\begin{equation}\begin{alignedat}{1}
    \displaystyle\frac{dV}{dt}&= \nabla V(x,g)^\top\cdot[\dot x\,\,\,\dot g]\\&=\nabla V(x,g)^\top\cdot F(x,g)+\nabla V(x,g)^\top\cdot G(x)\\
    &=\displaystyle 0+\frac{\partial V}{\partial x}\beta \bar x^{-1}x(1-x)(\bar x-x)\\
           &=\displaystyle \frac{x-\bar x}{\bar xx^2(1-x)}\beta \bar x(1-x)(\bar x-x)\\
           &=\displaystyle -\beta\frac{1}{x\bar x}(\bar x-x)^2\leq 0.
\end{alignedat}\end{equation}
Note that on the manifold $x=\bar x$, the only invariant set is the constant solution $(\bar x,\bar g)$. In fact, for $x=\bar x$ and $g\neq \bar g$, from the first equation in \eqref{eq:controlled_dominant_creation2} it holds $\dot x\neq 0$. Therefore, LaSalle's invariance principle~\cite{khalil2002nonlinear} guarantees asymptotic convergence to the desired equilibrium point.
\end{proof}

\begin{remark}
    The case in which $c>a$ and $d>b$ can be treated in a similar way  using our adaptive-gain controller with control matrix $\mat{G^{(2)}}$ and $\phi(x)=-p(x-\bar x)$.    
\end{remark}

\subsection{Anti-coordination games}

Now, we focus on anti-coordination games, which have an equilibrium $x^*=\frac{\beta}{\alpha+\beta}$ that is globally attractive from any initial condition in the interior (Proposition~\ref{prop:convergence_uncontrolled}). The presence of this equilibrium slightly complicates the design of the adaptive-gain controller. In particular, it requires us to have some additional information on whether such equilibrium is larger or smaller than the desired one.

\begin{assumption}\label{a:xstar}
    It is known whether the desired equilibrium $\bar x$ is such that $\bar x>x^*$ or $\bar x<x^*$.
\end{assumption}

Without any loss in generality, we assume that we want to steer the system to an equilibrium $\bar x<x^*$. In this case, we observe that our adaptive-gain controller should act similarly to what was discussed for a dominant strategy game in Section~\ref{sec:dominant_sp}. In fact, below the desired equilibrium, the uncontrolled dynamics (studied in Proposition~\ref{prop:convergence_uncontrolled}) has a constant drift in the direction of increasing $x$. Hence, the controller should (adaptively) increase the payoff associated with action~$2$, when  $x(t) > \bar x$.


\begin{theorem}\label{th:creation2}\rm
The adaptive-gain control $(\mat{G^{(3)}},\phi)$ with $\phi$ defined in \eqref{eq:linear_phi} solves Problem~\ref{pr:3} 
for an anti-coordination game with $\bar x<x^*$, for any initial condition  $x(0)\in(0,1)$. Precisely, the state converges to the desired equilibrium $\bar x$ and the gain converges to 
\begin{equation}
    \label{eq:g_bar5}
    \bar g=a+d-c-b+\frac{b-d}{\bar x}.
\end{equation}
\end{theorem}\medskip
\begin{proof}
By defining $\alpha=c-a>0$ and $\beta=b-d>0$, the controlled system with innovation-gain control ($\mat G=\mat{G^{(3)}}$) reads
\begin{equation}\label{eq:controlled_anticoordination_creation}
    \begin{alignedat}{1}
            \dot x&=x(1-x)(-(g+\alpha) x+\beta (1-x))\\&=\beta \bar x^{-1}x(1-x)(\bar x-x)-x^2(1-x)(g-\bar g)\\
            \dot g&=gp(x-\bar x),
    \end{alignedat}
\end{equation}
with $\bar g$ from \eqref{eq:g_bar5}. Notice that the condition $\bar x<x^*$ is necessary and sufficient to guarantee $\bar g>0$, i.e., that the desired equilibrium is a feasible point of the domain.  Similar to the proof of Theorem~\ref{th:creation}, we proceed by proving that \eqref{eq:lyapunov} is monotonically non-increasing along the trajectories of \eqref{eq:controlled_anticoordination_creation}, and strictly decreasing for $x\neq \bar x$. One the manifold $x=\bar x$, the unique (and thus largest) invariant set can be proved to be the desired equilibrium point $(x=\bar x, g=\bar g)$. Hence, LaSalle's invariance principle yields the claim. 
\end{proof}

\begin{remark}
    The case in which $\bar x>x^*$  can be treated in a similar way  using our adaptive-gain controller with control matrix $\mat{G^{(2)}}$ and $\phi(x)=-p(x-\bar x)$.    
\end{remark}

 Theorems~\ref{th:creation} and~\ref{th:creation2} guarantee that our adaptive-gain control scheme can be used to solve the set-point regulation problem for an anti-coordination or a dominant-strategy game. In other words, using our method, we can steer a replicator equation to any desired equilibrium point in a closed-loop fashion, without any a priori knowledge of the structure of the game (except for knowing the class of game). Figure~\ref{fig:case3} illustrates the theoretical findings in Theorem~\ref{th:creation}, comparing the performance of different choices of the parameter $p$, which regulates the velocity of the adaptation process. In particular, by comparing Fig.~\ref{fig:6a} and Fig.~\ref{fig:6b}, we observe how larger values of the parameter $p$ (which is increased from $p=0.1$ to $p=2$) guarantees faster convergence, but they generate transient oscillations about the desired equilibrium (yielding thus higher peaks in the gain), while with slower adaptation rates the convergence seems to be monotone. This interesting difference can be observed also from the state-space plots in Figs.~\ref{fig:6e} and~\ref{fig:6f}.

Finally, it is worth noticing that Theorem~\ref{th:creation} provides a sufficient condition for the adaptation rate to guarantee the solution of Problem~\ref{pr:3}, i.e., to have the linear form in \eqref{eq:linear_phi}. Intuitively, more general expressions of the adaptation rate obtain by composing \eqref{eq:linear_phi} with functions that preserves the sign may be adopted. However, the general analytical treatment of such scenarios is nontrivial, since the Lyapunov-like function used in the proof of Theorem~\ref{th:creation} cannot be directly employed. Moreover, numerical simulations suggest a nontrivial dependence of the outcome on the adaptation rate. In fact, while the use of a saturation term may hasten the convergence (see, e.g., Figs.~\ref{fig:6c} and~\ref{fig:6g} where we use $\phi(x)=\text{atan}(x-\bar x)$), nonlinear functions that wane too fast as $x(t)$ approaches the desired equilibrium may fail to yield convergence, as shown in Figs.~\ref{fig:6d} and~\ref{fig:6h} where $\phi(x)=(x-\bar x)^3$ produces an offset between the desired equilibrium and the asymptotic state of the system. The results of our numerical simulations suggest that the generalization of our theoretical results to nonlinear adaptation rates is nontrivial.


\section{Conclusion}\label{sec:conclusions}

In this paper, we dealt with the equilibrium selection problem for the replicator equation by designing a novel adaptive-gain control scheme that requires limited a priori information on the structure of the game. Specifically, we proposed a control scheme that uses information at the population level to adaptively change the individual-level utilities associated with the game played by the individuals. We demonstrated that the proposed closed-loop controller is able to successfully steer a population to any locally (but not globally) stable equilibrium and to any unstable consensus state, establishing easy-to-implement sufficient conditions to solve the equilibrium selection problem, with applications spanning from promoting social change to favoring cooperation in social dilemmas~\cite{kreindler2014rapid_diffusion,ye2021nat,Stella2022cooperation}. Moreover, we demonstrated how the proposed controller can be tailored to steer the system to a state that is not even an equilibrium of the uncontrolled replicator equation, with potential applications in congestion problems for infrastructure networks~\cite{Jiang2014,Como2022traffic}.

Our novel control scheme paves the way for several lines of future research. First, this paper provides sufficient conditions for solving the equilibrium selection problem. However, some open questions remain. In particular, future research should investigate to what extent the sufficient conditions on the design of the adaptation rate can be relaxed. Especially, for the set-point regulation problem, numerical simulations suggest that the extension to general nonlinear adaptation rates may be nontrivial and requires the development of different theoretical tools than the Lyapunov-like functions used in Theorem~\ref{th:creation}, such as passivity-based approaches~\cite{Park2019}. Furthermore, future efforts should focus on extending the control method to solve the set-point regulation for coordination games, and to deal with more general classes of revision protocols, such as imitation and best-response dynamics~\cite{Sandholm2010}, and populations characterized by heterogeneity and a network structure~\cite{Madeo2015,Como2021,ramazi2016networks,Le2024}. Second, our theoretical results establish that the equilibrium selection problem can be solved using an adaptive-gain control scheme with a wide range of different adaptation rates. Numerical simulations suggest that the choice of tuning parameters substantially impact convergence speed and control cost. This outlines a promising line of  research towards optimal design of the control parameters. Third, in this paper we presented different real-world scenarios that can be formalized as an equilibrium selection problem (e.g., promoting social change or cooperation in social dilemmas) and we discussed how the control policy obtained using our adaptive-gain controller can be mapped into real actions that public authorities can enact. Future efforts should  focus on exploring the feasibility of such real-world implementations, extending the mathematical framework to account for noise and control actions that can be enacted in discrete-time (e.g., in an event-triggered fashion), toward closing the loop between theoretical research and real-world applications. \bigskip

{\bf Acknowledgments. }This work was partially supported by the Western Australian Government (Premier's Science Fellowship Program) and was carried out within the FAIR---Future Artificial Intelligence Research and received funding from the European Union Next-GenerationEU (PIANO NAZIONALE DI RIPRESA E RESILIENZA (PNRR) – MISSIONE 4 COMPONENTE 2, INVESTIMENTO 1.3 – D.D. 1555 11/10/2022, PE00000013). This manuscript reflects only the authors’ views and opinions, neither the European Union nor the European Commission can be considered responsible for them.


\end{document}